\documentclass[acmlarge,11pt]{acmart}
\usepackage{booktabs}
\usepackage{amsmath}
\usepackage{amsthm}
\usepackage{algorithm}
\usepackage{algpseudocode}
\algtext*{EndWhile}
\algtext*{EndIf}
\usepackage{listings}
\usepackage{multicol}
\usepackage{xspace}
\usepackage{ifthen}
\usepackage{geometry}
\usepackage{xcolor}
\usepackage[moderate,margins]{savetrees}
\usepackage[inline]{enumitem}
\usepackage[font=small,skip=6pt,margin=12pt]{caption}
\usepackage[abbreviations]{foreign}  
\usepackage{wrapfig}
\usepackage{xpatch}

\makeatletter
\xpatchcmd{\algorithmic}{\itemsep\z@}{\itemsep=-1pt}{}{}
\makeatother

\newboolean{showcomments}
\setboolean{showcomments}{true}
\ifthenelse{\boolean{showcomments}}
{ \newcommand{\mynote}[3]{
    \fbox{\bfseries\sffamily\scriptsize#1}
    {\small$\blacktriangleright$\textsf{\emph{\color{#3}{#2}}}$\;\blacktriangleleft$}}}
{ \newcommand{\mynote}[3]{}}

\newcommand{\CX}{\textsf{CX}\xspace}
\newcommand{\CXT}{\textsf{CXTimed}\xspace}

\begin{document}

\setcopyright{none}

\fancypagestyle{firstpagestyle}{\fancyhead{}\fancyfoot{}}
\fancyhead{}\fancyfoot{}

\sloppy

\newtheoremstyle{compactbf}%
{2pt}
{1pt}
{\itshape}
{}
{\textcolor{darkgray}{$\blacktriangleright$}\nobreakspace\sffamily\bfseries}
{.}
{.5em}
{}

\theoremstyle{compactbf}
\newtheorem{thm}{Theorem}
\newtheorem{defi}{Definition}
\newtheorem{prop}{Proposition}
\newtheorem{lem}{Lemma}
\newtheorem{cor}{Corollary}

\newcommand{\thmautorefname}{Theorem}
\newcommand{\defiautorefname}{Definition}
\newcommand{\propautorefname}{Proposition}
\newcommand{\lemautorefname}{Lemma}
\newcommand{\corautorefname}{Corollary}
\newcommand{\algorithmautorefname}{Algorithm}

\algrenewcommand\alglinenumber[1]{\tiny #1:}
\algrenewcommand\algorithmicindent{1.0em}

\lstset{basicstyle=\linespread{0.6}\footnotesize\sffamily,columns=fullflexible,language=C++,numbers=left,stepnumber=1,numberstyle=\tiny,numbersep=3pt,numberblanklines=false,escapeinside={(*@}{@*)}}
\def\ContinueLineNumber{\lstset{firstnumber=last}}
\def\StartLineAt#1{\lstset{firstnumber=#1}}
\let\numberLineAt\StartLineAt


\setlength{\textfloatsep}{2pt}
\setlength{\intextsep}{2pt}


\title{A Wait-Free Universal Construct for Large Objects}

\author{Andreia Correia}
\affiliation{\institution{University of Neuchatel}}
\email{andreia.veiga@unine.ch}

\author{Pedro Ramalhete}
\affiliation{\institution{Cisco Systems}}
\email{pramalhe@gmail.com}

\author{Pascal Felber}
\affiliation{\institution{University of Neuchatel}}
\email{pascal.felber@unine.ch}

\renewcommand{\shorttitle}{CX Universal Construct}


\begin{abstract}
Concurrency has been a subject of study for more than 50 years. 
Still, many developers struggle to adapt their sequential code to be accessed concurrently.
This need has pushed for generic solutions and specific concurrent data structures.

Wait-free universal constructs are attractive as they can turn a sequential implementation of any object into an equivalent, yet concurrent and wait-free, implementation. 
While highly relevant from a research perspective, these techniques are of limited practical use when the underlying object or data structure is sizable. 
The copy operation can consume much of the CPU's resources and significantly degrade performance.

To overcome this limitation, we have designed \CX, a multi-instance-based wait-free universal construct that substantially reduces the amount of copy operations.
The construct maintains a bounded number of instances of the object that can potentially be brought up to date.
We applied \CX to several sequential implementations of data structures, including STL implementations, and compared them with existing wait-free constructs.
Our evaluation shows that \CX performs significantly better in most experiments, and can even rival with hand-written lock-free and wait-free data structures, simultaneously providing wait-free progress, safe memory reclamation and high reader scalability. 


\end{abstract}

%
%



\maketitle


\section{Introduction}
\label{sec:introduction}

Many synchronization primitives have been proposed in the literature for providing concurrent access to shared data, with the two most common being mutual exclusion locks and reader-writer locks.
Both of these primitives provide blocking progress, with some mutual exclusion algorithms like the ticket lock~\cite{mellor1991algorithms} or CLH lock~\cite{magnusson1994queue}
going so far as being starvation free. 
Even today, the usage of locks is still of great relevance because of their generality and ease of use, despite being prone to various issues such as \emph{priority inversion}, \emph{convoying} or \emph{deadlock}~\cite{herlihy1993transactional}. 
Yet, their main drawback comes from their lack of scalability and suboptimal use of the processing capacity of multi-core systems, except for lock-based techniques with disjoint access.
This has led researchers to extensively explore alternatives to support non-blocking data structures, either using ad-hoc algorithms or generic approaches.

Generic constructs are attractive from a theoretical perspective, but so far they have been largely neglected by practitioners because of their lack of efficiency when compared to dedicated algorithms tailored for a specific data structure.
The search for a generic non-blocking solution that is also practical has resulted in significant developments over the last decades, notably in the fields of wait-free universal constructs (UCs) and hardware and software transactional memory (HTM and STM).

A wait-free universal construct is a generic mechanism meant to provide concurrent wait-free access to a sequential implementation of an object or group of objects, \eg, a data structure. 
In other words, it takes a \emph{sequential specification} of an object and provides a concurrent implementation with wait-free progress~\cite{Herlihy91}. 
It supports an operation, called \texttt{applyOp()}, which takes as a parameter the sequential implementation of any operation on the object, and simulates its execution in a concurrent environment.
Most UCs can be adapted to provide an API that distinguishes between read-only operations and mutative operations on the object, which henceforth will be referred to as \texttt{applyRead()} and \texttt{applyUpdate()} respectively.

Software transactional memory, on the other hand, has transactional semantics, allowing the user to make an operation or group of operations seem \emph{atomic} and providing \emph{serializability} between transactions~\cite{herlihy1993transactional}. 
STMs and UCs present two separate approaches to developers when it comes to dealing with concurrent code.
Both approaches allow the end user to reason about the code as if it were sequential.
STMs \emph{instrument} the loads and stores on the sequential implementation.
STMs may also require type annotation, function annotation or replacement of allocation/deallocation calls with equivalent methods provided by the STM.
Finally, to the best of our knowledge, there is currently no STM with wait-free progress.
A recent development, named RomulusLR~\cite{correia2018romulus}, provides wait-free progress for read-only operations and blocking starvation-free updates.

We observe that the UC literature can be classified into two groups of algorithms, UCs that do not require instrumentation, and UCs that do.
\emph{Non-instrumenting} UCs require no annotation of the sequential implementation, allowing the developer to \emph{wrap} the underlying object and creating an equivalent object with concurrent access for all of its methods.
\emph{Instrumenting} UCs require the developer to annotate and modify the sequential implementation, similar to what must be done for an STM.
This annotation implies effort from the developer and is prone to errors.
In addition, the fact that annotation is required at all, makes it difficult or unfeasible to use legacy code or data structures provided by pre-compiled libraries (\eg, \texttt{std::set} and \texttt{std::map}) because it would require modifying the library's source code.

In this paper, we focus on non-instrumenting UCs with the goal of addressing their main limitations in terms of performance and usability, which made them so far impractical for real-world applications.
We introduce \CX, a non-instrumenting UC with linearizable operations and wait-free progress that does not require any annotation of the underlying sequential implementation.
\CX provides fast and scalable read-only operations by exploiting their \emph{disjoint access parallel}~\cite{israeli1994disjoint} nature.

In short, with \CX we make the following contributions:
\begin{enumerate*}[label=\emph{(\roman*)}]
\item We introduce the first practical wait-free UC, written in portable C++, with integrated wait-free memory reclamation and high scalability for read-mostly workloads.
\item We address wait-free memory reclamation with a flexible scheme that combines reference counting with hazard pointers
\item We present the first portable implementation of the \textsf{PSim} UC~\cite{fatourou2011highly}, with integrated wait-free memory reclamation, and added high scalability for read-only operations.
\end{enumerate*}

The rest of the paper is organized as follows.
We first discuss related work in \S\ref{sec:related}.
We then present the \CX algorithm in \S\ref{sec:algorithm}.
We perform an in-depth evaluation of \CX in \S\ref{sec:evaluation} and finally conclude in \S\ref{sec:conclusion}.
Proofs of correctness are presented in Appendix.

\section{Related Work}
\label{sec:related}

In 1983, \citet{peterson1983concurrent} was the first to attack the problem of non-blocking access to shared data and to provide several solutions to what he called the \emph{concurrent reading while writing} problem.
One of these solutions uses two instances of the same data and guarantees wait-free progress for both reads and writes, allowing multiple readers and a single writer to access simultaneously any of the two instances.
However, this approach is based on \emph{optimistic concurrency}, causing read-write races, which has troublesome implications in terms of atomicity, memory reclamation and invariance conservation.

Later, in 1990, \citet{maurice1990methodology} proposed the first wait-free UC for any number of threads. 
This UC requires no annotation or modification to the sequential implementations and is therefore a non-instrumenting UC.
His approach keeps a list of all operations ever applied, and for every new operation it will re-apply all previous operations starting from an instance in its initial state.
One by one, as each operation is appended, the list of operations grows unbounded until it exhausts all available memory, thus making this UC unsuitable for practical usage.

Since then, several wait-free UCs have been proposed~\cite{anderson1995universal,anderson1995universalmulti,chuong2010universal,ellen2016universal,fatourou2011highly}. 
Researchers have attempted to address the problem of applying wait-free UCs to large objects~\cite{afek1995wait,anderson1995universal,maurice1990methodology} though none has succeeded in providing a generic solution~\cite{raynal2017distributed}.

\citet{anderson1995universal} have proposed a technique designed to work well with large objects.
This algorithm is an instrumenting UC because it requires the end user to write a sequential procedure that treats the object as if it was stored in a contiguous array, implying adaptation or annotation of the sequential implementation.
This technique is also not universally applicable to all data structures. 

\citet{chuong2010universal} have shown a technique that makes a copy of each shared variable in the sequential implementation and executes the operations on the copy.
Although this technique can operate at the level of memory words, it would be vulnerable to race conditions from different (consecutive) operations that modify the same variables.
Even if a CAS would be used to modify these variables, ABA issues could still occur.
This word-based approach requires instrumentation of the user code.

\citet{fatourou2011highly} have designed and implemented \textsf{P-Sim}, a highly efficient wait-free and non-instrumenting UC based on fetch-and-add and LL/SC.
\textsf{P-Sim} relies on an up-to-date instance that all threads copy from, using Herlihy's combining consensus~\cite{herlihy1993methodology} to establish which operations are to be applied on the copy, independently of whether these are read or write operations.
In the best-case scenario, one copy of the entire object state is made per $N$ concurrent operations, where $N$ is the number of threads.
In the worst-case, two copies are done per operation.
Unfortunately, \textsf{P-Sim} is impractical for large objects.

More recently, \citet{ellen2016universal} have shown a wait-free UC based on LL/SC. 
Their technique provides disjoint access parallel operations with the requirement that all data items in the sequential code are only accessed via the instructions \texttt{CreateDI}, \texttt{ReadDI} and \texttt{WriteDI}, implying the need for instrumentation of the sequential implementations similar to an STM.
No implementation has been made publicly available.


\section{CX Algorithm}
\label{sec:algorithm}

\newcommand{\COMB}[1]{\ensuremath{\mathit{Comb}_{#1}}\xspace}
\newcommand{\CURCOMB}{\ensuremath{\mathit{curComb}}\xspace}
\newcommand{\HEAD}[1]{\ensuremath{\mathit{head}_{#1}}\xspace}


The \CX wait-free construct uses a wait-free queue where mutations to the object instance are placed, much like Herlihy's wait-free construct, though instead of each thread having its own copy of the instance, there are a limited number of copies that all threads can access.
The access to each of these copies is protected by a reader-writer lock, which can be acquired by multiple reader threads in \emph{shared mode}, whereas only one writer thread can get the lock in \emph{exclusive mode}.
The reader-writer lock used in \CX must guarantee that, when multiple threads compete for the lock using the \texttt{trylock()} method, at least one will succeed and obtain the lock.
This property, named \emph{strong trylock}~\cite{correia2018strong}, combines \emph{deadlock freedom} with linearizable consistency and wait-free progress. 

The \textsf{CX} construct (see \autoref{fig:examples}) is composed of:
\begin{enumerate*}[label=\emph{(\roman*)}]
\item \texttt{curComb}: a pointer to the current \texttt{Combined} instance;
\item \texttt{tail}: a pointer to the last node of the queue; and
\item \texttt{combs}: an array of \texttt{Combined} instances.
\end{enumerate*}

In turn, a \texttt{Combined} instance consists of:
\begin{enumerate*}[label=\emph{(\roman*)}]
\item \texttt{head}: a pointer to a \texttt{Node} on the queue of mutations;
\item \texttt{obj}: a copy of the data structure or object that is up to date until \texttt{head}, \ie, any mutative operation that was enqueued after \texttt{head} has not yet been applied to \texttt{obj}; and
\item \texttt{rwlock}: an instance of a reader-writer lock that protects the content of the \texttt{Combined} instance.
\end{enumerate*}

Finally, a \texttt{Node} holds:
\begin{enumerate*}[label=\emph{(\roman*)}]
\item \texttt{mutation}: a function to be applied on the object;
\item \texttt{result}: the value returned by the update function, if any;
\item \texttt{next}: a pointer to the next \texttt{Node} in the mutation queue;
\item \texttt{ticket}: a sequence number to simplify the validation in case of multiple threads applying the same mutations;
\item \texttt{refcnt}: a reference counter for memory reclamation, as well as some other fields for internal use.
\end{enumerate*}
The definitions for the main data structures of \CX are shown in Algorithm~\ref{alg:class}.

\begin{algorithm}[ht]
\vspace{-10pt}
\caption{\CX data structures}\label{alg:class}
\begin{lstlisting}[multicols=2]
template<typename C, typename R = uint64_t> 
class CX { (*@\hfill@*) // C = "Object type" ; R = "Result type"
  const int maxThr;
  std::atomic<Combined*> curComb {nullptr};
  std::function<R(C*)> mut0 = [](C* c) {return R{};};
  Node* sentinel = new Node(mut0, 0);
  std::atomic<Node*> tail {sentinel};
  Combined* combs;
  struct Node {
    std::function<R(C*)>       mutation;    (*@\hfill@*) // Not a pointer (*@\hspace{3pt}\columnbreak@*)
    std::atomic<R>             result;      (*@\hfill@*) // Not a pointer
    std::atomic<Node*>         next {nullptr};
    std::atomic<uint64_t>      ticket {0};
    const int                  enqTid;     (*@\hfill@*) // Used internally by queue
  };
  struct Combined {
    Node*  head {nullptr};
    C*  obj {nullptr};
    StrongTryRWRI rwLock {maxThr};
  };};
\end{lstlisting}
\vspace{-8pt}
\end{algorithm}

\autoref{fig:examples} illustrates the data structures and principle of \CX on a concurrent stack.
Mutative operations in the wait-free queue are represented by rounded rectangles, with node \textsf{A} corresponding to operation \texttt{push(a)}.
The stack stores its elements in a linked list of nodes (circles), with dashed lines indicating the nodes that are not yet added to the specific instance of the data structure.

The \CX construct relies on the copies of the object present in the \texttt{combs[]} array.
Initially, one of the \texttt{Combined} instance in the array holds an initialized object \texttt{obj} and its \texttt{head} pointer refers to the sentinel node $\bot$ of \CX's wait-free queue of mutations.
The other instances have both \texttt{obj} and \texttt{head} set to null.

To improve \emph{readers} performance, \CX has distinct code paths for \emph{readers} and \emph{updaters}. 
Readers call \texttt{applyRead()}, which tries to acquire the shared lock on the reader-writer lock instance of the current \texttt{Combined} instance, \texttt{curComb}.
Updaters call \texttt{applyUpdate()}, which scans the \texttt{combs[]} array and attempts to acquire the exclusive lock on the reader-writer lock of one of the \texttt{Combined} instances (which is guaranteed to succeed after a maximum of \texttt{numReaders}+2$\times$\texttt{numUpdaters} trials).

An updater thread has to account for a possible read operation that will be executed in the most up-to-date copy, which is referenced by \texttt{curComb}. 
The updater is responsible for leaving \texttt{curComb} referencing a copy that contains its mutative operation, and this copy is left with a shared lock held so as to protect it from being acquired in exclusive mode by updater threads, including itself.
When a copy of the object is required, the copy procedure will try to acquire \texttt{curComb} in shared mode and copy its \texttt{obj} while holding the shared lock, therefore guaranteeing a consistent replica. 
This implies that two \texttt{Combined} instances may be required for each updater thread: the original copy and the new replica.
If we consider that the construct will be accessed by at most \texttt{numReaders} dedicated readers and \texttt{numUpdaters} dedicated updaters, then the maximum number of \texttt{Combined} instances in use at any given time will be one per reader plus two per updater, \ie, \texttt{numReaders}+2$\times$\texttt{numUpdaters} (2$\times$\texttt{maxThreads} if every thread can potentially update the data structure).

Once an updater thread secures a \texttt{Combined} instance with exclusive access and ensures it has a copy of the object, the updater is responsible for applying all the mutations present in the mutation queue following the sequential order from the \texttt{head} of the \texttt{Combined} instance until its own mutative operation, which was previously added to the queue of mutations. 
Each node has a \texttt{ticket} that simplifies the validation in case the mutation has already been applied and \texttt{head} is more recent than the node $N$ containing the mutative operation to be made \emph{visible}.
The concept of visibility refers to a state of the object, where effects of operations on the object are available to all threads.
After the \texttt{Combined} instance is brought up-to-date with a copy of the object containing the updater's mutation, the updater thread has to make its mutation visible to other threads by ensuring that \texttt{curComb} advances to a \texttt{Combined} instance whose \texttt{head} has a \texttt{ticket} greater than or equal to $N$'s \texttt{ticket}.

We define a valid copy of the object as an instance that can be brought up to date, applying all the mutations starting from the \texttt{head} of the \texttt{Combined} where the copy is stored, until $N$'s mutation. 
An invalidation of a copy occurs when there is memory reclamation of the queue's nodes. In case a copy is invalidated, a new copy can be created from \texttt{curComb}.

\begin{figure*}
  \centering
  \includegraphics[scale=0.65]{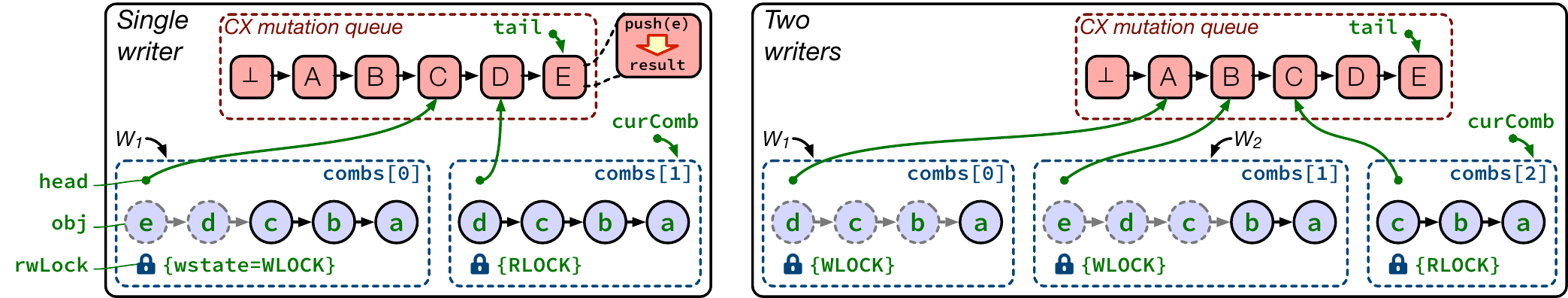}
  \caption{
    Illustration of \CX's principle on two scenarios with one (left) and two (right) writers pushing elements \texttt{a} through \texttt{e} in a shared stack.
  }
  \label{fig:examples}
\end{figure*}

Consider \autoref{fig:examples} to better understand how mutations are propagated to the available copies.
In the left figure, a single writer $W_1$ has been pushing values \texttt{a}, \texttt{b}, \texttt{c} and \texttt{d}, and is in the process of executing \texttt{push(e)}.
The mutative operation has already been inserted at the tail of \CX's wait-free queue but not yet applied to the stack.
At that point \texttt{curComb} points to \texttt{combs[1]}, which holds an up-to-date stack (with all 4 elements inserted) protected by a shared lock.
Hence the writer cannot use this instance and instead acquires \texttt{combs[0]} in exclusive mode.
The next steps for the writer will be to apply the operations starting from the \texttt{head}, \ie, push \texttt{d} and \texttt{e}, to bring the data structure up to date, update \texttt{head} to point to the last applied mutation (node \textsf{E}), atomically set \texttt{curComb} to point to \texttt{combs[0]}, and finally downgrade the lock to read mode.
The figure on the right presents a similar scenario but with two concurrent writers $W_1$ and  $W_2$, with the first one executing \texttt{push(d)} on \texttt{combs[0]} and the second one \texttt{push(e)} on \texttt{combs[1]}.
The order of operations is determined by their position in the wait-free queue, i.e., \texttt{d} is inserted before \texttt{e} and both writers will apply the operations on the \texttt{Combined} instances in this order.


\subsection{Algorithm Walkthrough}

\begin{algorithm}[ht!]
	\caption{CX ApplyRead and ApplyUpdate pseudo-code}\label{alg:cxread}\label{alg:cxmut}\vspace{-15pt}
	\scriptsize\linespread{0.6}
    \begin{multicols}{2}
  	\begin{algorithmic}[1]
		\Function {applyRead}{readFunc,tid}:
		\For {\_i $\gets$ 0, MAX\_READ\_TRIES + MAX\_THREADS}
		\If{\_i $=$ MAX\_READ\_TRIES}
		\State \_myNode $\gets$ enqueue(updFunc) \label{alg:cxread:enqueue-node}
		\EndIf
		\State \_comb $\gets$ curComb \label{alg:cxread:read:curcomb-load}
		\If{\_comb.rwlock.sharedTryLock(tid)} \label{alg:cxread:read:shared-lock}
		\If {\_comb == curComb}
		\State \_ret $\gets$ readFunc() \hfill $\triangleright$ \emph{Function call}
		\State \_comb.rwlock.sharedUnlock(tid)
		\State \textbf{return} \_ret
		\EndIf
		\State \_comb.rwlock.sharedUnlock(tid)
		\EndIf
		\EndFor \hfill $\triangleright$ \emph{A writer must have completed its operation$\ldots$}
		\State \textbf{return} \_myNode.result \hfill $\triangleright$ \emph{$\ldots$and the result is in} myNode
		\EndFunction
	\vspace{5pt}
		\Function {applyUpdate}{updFunc}:
		\State \_myNode $\gets$ enqueue(updFunc) \label{alg:cxmut:create-node} \label{alg:cxmut:enqueue-node} \label{alg:cxmut:set-ticket} \hfill \{\,1\,\}
		\State \_tkt $\gets$ myNode.ticket
		\State \_c, \_idx $\gets$ exclusiveTryLock() \label{alg:cxmut:get-instance:start} \label{alg:cxmut:lock-exclusive} \label{alg:cxmut:get-instance:end} \hfill \{\,2\,\}
		\State \_mn $\gets$ \_c.head \label{alg:cxmut:apply:first-node}
		\If {\_mn $\neq$ null $\wedge$ \_mn.ticket $\geq$ \_tkt} \label{alg:cxmut:compare-ticket-head}
		\State \_c.rwLock.exclusiveUnlock()
		\State \textbf{return} \_myNode.result \label{alg:cxmut:return-result-1}
		\EndIf
		\State \_comb $\gets$ \textbf{null}
		\State \_combIdx $\gets$ -1
		\While {$\_mn \neq \_myNode$} \label{alg:cxmut:apply:start} \label{alg:cxmut:apply:last-node}
		\If {\_mn $=$ \textbf{null} $\vee$ \_mn $=$ \_mn.next} \label{alg:cxmut:copy-required} \label{alg:cxmut:copy:start} \hfill \{\,3\,\}
		\State \_combIdx $\gets$ getCombined(\_tkt) \label{alg:cxmut:acquire-shared}
    \columnbreak
		\If {\_comb $\neq$ \textbf{null} $\vee$ \_combIdx $=$ -1} 
		\State \_c.head $\gets$ \_mn
		\State \_c.rwLock.exclusiveUnlock()
		\State \textbf{return} \_myNode.result \label{alg:cxmut:return-result-2}
		\EndIf
		\State \_comb $\gets$ combs[\_combIdx]
		\State \_mn $\gets$ \_comb.head
		\State \_c.updateHeadObj(\_comb, \_mn) \label{alg:cxmut:update-head-1} \label{alg:cxmut:delete-newcomb} \label{alg:cxmut:copy-obj}
    	\State \_comb.rwLock.sharedUnlock()
		\State \textbf{continue} \label{alg:cxmut:copy-completed}
		\EndIf \label{alg:cxmut:copy:end}
		\State \_mn $\gets$ \_mn.next
		\State \_mn.result.store(\_mn.updFunc(\_comb.obj)) \label{alg:cxmut:store-result} \hfill \{\,4\,\}
		\EndWhile \label{alg:cxmut:apply:end}
		\State \_c.head $\gets$ \_mn \label{alg:cxmut:update-head-2} \hfill \{\,4\,\}
		\State \_c.rwLock.downgradeToHandover() \label{alg:cxmut:downgrade-to-shared-lock} \hfill \{\,5\,\}
		\For {\_i $\gets$ 0, MAX\_THREADS}
		\State \_combIdx $\gets$ curComb
		\State \_comb $\gets$ combs[\_combIdx]
		\If{$\neg$sharedTryLockCheckTkt(\_comb)} \label{alg:cxmut:lock-shared} \label{alg:cxmut:compare-ticket}
		\State \textbf{continue}
		\EndIf
		\If {curComb.cas(\_combIdx, \_idx)} \label{alg:cxmut:cas} \hfill \{\,6\,\}
		\State \_comb.rwLock.handoverUnlock()
		\State \_comb.rwLock.sharedUnlock()
		\State \textbf{return} \_myNode.result \label{alg:cxmut:return-result-3}
		\EndIf
		\State \_comb.rwLock.sharedUnlock()
		\EndFor
		\State \_c.rwLock.handoverUnlock() \label{alg:cxmut:final-unlock}
		\State \textbf{return} \_myNode.result \label{alg:cxmut:return-result-4}
		\EndFunction
	\end{algorithmic}
	\end{multicols}\vspace{-10pt}
\end{algorithm}


The core of the \CX algorithm resides in the \texttt{applyUpdate()} mutative operation, shown in \autoref{alg:cxmut}. The main steps of the algorithms are:
\begin{enumerate}[leftmargin=*, align=left]
\setlength\itemsep{0em}
\item Create a new \texttt{Node} \texttt{\_myNode} with the desired mutation and insert it in the queue (line~\ref{alg:cxmut:create-node}).
\item Acquire an exclusive lock on one of the \texttt{Combined} instances in the \texttt{combs[]} array (line~\ref{alg:cxmut:get-instance:start}).
\item Verify if there is a valid copy of the data structure in \texttt{\_c} and make a copy if necessary (line~\ref{alg:cxmut:copy:start}).
\item Apply all mutations starting at \texttt{head} of the \texttt{Combined} instance until reaching the \texttt{Node} inserted in the first step (lines~\ref{alg:cxmut:apply:start} to~\ref{alg:cxmut:apply:end}), and update \texttt{head} to point to this node (line~\ref{alg:cxmut:update-head-2}).
\item Downgrade lock on \texttt{\_c} (line~\ref{alg:cxmut:downgrade-to-shared-lock}).
\item Compare-and-set (CAS) \texttt{curComb} from its current value to the just updated \texttt{Combined} instance (line~\ref{alg:cxmut:cas}). 
Upon failure, retry CAS until successful or until \texttt{head} of the current \texttt{curComb} instance is \emph{after} \texttt{\_myNode}.
\end{enumerate}

When applying a mutation to the underlying object, the first step is to create a new node with the mutation (line~\ref{alg:cxmut:create-node} of \autoref{alg:cxmut}) and insert it in \CX's queue.
Each node contains a \texttt{mutation} field that stores the mutation. 
A monotonically increasing \texttt{ticket} is assigned to the node to uniquely identify the mutation (line~\ref{alg:cxmut:set-ticket}).

The next step consists in finding an available \texttt{Combined} instance on which to apply the new mutation.
To that end, the thread must acquire a \texttt{Combined}'s lock in exclusive mode (line~\ref{alg:cxmut:lock-exclusive}). 
The \textsf{StrongTryRWRI}~\cite{correia2018strong} reader-writer lock provides a \emph{strong} \texttt{exclusiveTryLock()} method, guaranteeing that the lock will be acquired in at most 2$\times$\texttt{maxThreads} attempts.

If the locked \texttt{Combined} instance has an invalidated or null \texttt{obj} (line~\ref{alg:cxmut:copy:start}), we need to make a copy of the current object. 
To do so, we first acquire the shared lock on \texttt{curComb} (line~\ref{alg:cxmut:acquire-shared}), before updating \texttt{head} and copying \texttt{obj} (line~\ref{alg:cxmut:copy-obj}).
It is worth mentioning that copy-on-write (COW) based techniques usually make one such copy for every mutative operation, while \CX does this \emph{once} for every new used \texttt{Combined} (of which there are 2$\times$\texttt{maxThreads}) plus the number of times a copy is invalidated.

Next, we apply the mutations on \texttt{obj}, starting from the corresponding \texttt{head} node (line~\ref{alg:cxmut:apply:first-node}) until our newly added node is found (line~\ref{alg:cxmut:apply:last-node}), always saving the result of each mutation in the corresponding \texttt{node.result} (line~\ref{alg:cxmut:store-result}). 
The rationale for saving the result is that, if another thread calling \texttt{applyUpdate()} sees that its own mutation is already visible at \texttt{curComb}, it can directly return the result of the mutation (lines~\ref{alg:cxmut:return-result-1}, \ref{alg:cxmut:return-result-2}, \ref{alg:cxmut:return-result-3} and~\ref{alg:cxmut:return-result-4}) without having to actually execute the previous mutations in the queue.
This approach implies that multiple threads may be write-racing the same value into \texttt{node.result}, which means that it must be accessed atomically, further implying that \texttt{R} must fit in a \texttt{std::atomic} data type to ensure wait freedom.

After the mutations have been applied, we downgrade the lock on \texttt{\_c} and advance \texttt{curComb} with a CAS (line~\ref{alg:cxmut:cas}) so as to make the current and previous mutations visible to other threads: \texttt{curComb} will now reference a \texttt{Combined} instance that contains the effects until \texttt{head.ticket}, also \texttt{curComb} will be protected by a shared lock to guarantee the instance is always available to execute a read operation.
In addition, \texttt{curComb} always transitions to a \texttt{Combined} instance with a \texttt{head.ticket} higher than the previous one thanks to the test on line~\ref{alg:cxmut:compare-ticket}.
This guarantees that operations whose effects are visible on \texttt{curComb} will remain visible.
Finally, we can unlock the unneeded \texttt{Combined} instance to make it available for exclusive locking by other threads searching for their \texttt{\_c}, before returning the result of our mutation.


As mentioned previously, a read operation will attempt to acquire a shared lock in the most up to date instance \texttt{curComb}. 
The reader may be unsuccessful in acquiring the lock if an updater has already acquired the lock. 
Such a situation can occur if, between the load on line~\ref{alg:cxread:read:curcomb-load} and the call to \texttt{sharedTryLock()} on line~\ref{alg:cxread:read:shared-lock}, the \texttt{curComb} advances to another \texttt{Combined} instance and an updater takes the exclusive lock on the previous instance.
This results in a lock-free progress condition for read operations, as it would only fail to progress in case an updater had made progress.
To guarantee wait-free progress, the reader must publish its operation in the queue after \texttt{MAX\_READ\_TRIES} attempts (line~\ref{alg:cxread:enqueue-node}), but it is not required to apply all mutations up to its own operation. 
After a maximum of \texttt{maxThread} transitions of \texttt{curComb} the reader's operation will be processed by an updater thread and become visible.
In case there is no updater thread to process the read operation, this implies that there is also no updater thread to block the reader thread from acquiring the lock on \texttt{curComb}. 

\begin{figure*}
	\centering
	\includegraphics[scale=0.65]{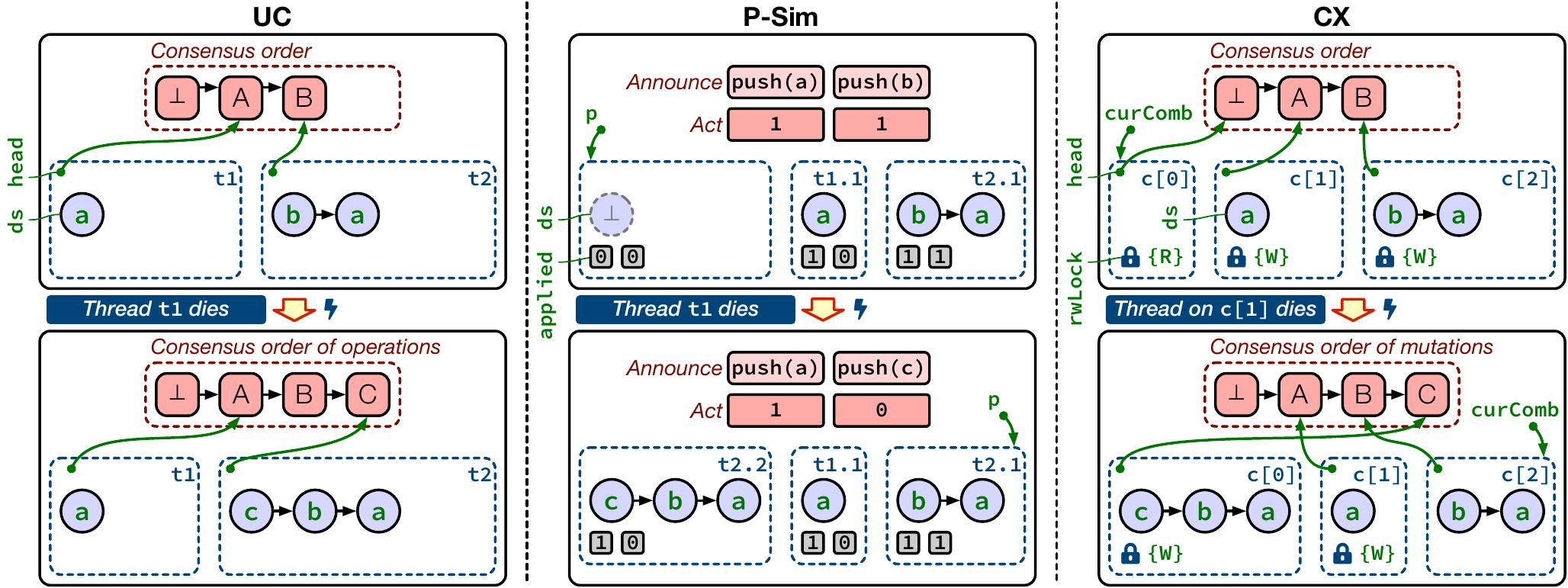}
	\caption{
		Comparison of UC~\cite{maurice1990methodology}, \textsf{P-Sim}~\cite{fatourou2011highly} and \CX.
	}
	\label{fig:comparison}
\end{figure*}

To better understand the differences between Herlihy UC~\cite{maurice1990methodology}, P-Sim~\cite{fatourou2011highly} and \CX, we show in \autoref{fig:comparison} a comparison example where two writer threads are pushing elements to a stack.
Both threads reach a consensus where the operation \texttt{push(a)} will be executed followed by \texttt{push(b)}.
Thread $T_1$ terminates abruptly while executing operation \texttt{push(a)}, while thread $T_2$ continues execution pushing element \texttt{c} to the stack.
With Herlihy's UC, every thread maintains its own copy of the data structure and each copy has to know the exact order in which the operations must be executed.
Because $T_2$ has no way to know if $T_1$ has died or is just delayed, the order of the operation from $A$ onwards has to be kept, which will cause it to grow indefinitely, eventually exhausting the system's memory.

P-Sim uses a copy-on-write approach with Herlihy's combining consensus~\cite{herlihy1993methodology}.
Each thread performs a copy of the object referenced by $P$ and applies to its copy all operations newly published in the \emph{announce} array.
There is no consensus order that all threads agree upon; instead, the thread that is able to transition $P$ to its copy is the one establishing the order of the operations.
In case thread $T_1$ dies, its copy (\texttt{t1.1}) is left unreclaimed but this has no impact on the execution of thread $T_2$.

CX has a pool of 2$\times$\texttt{maxThreads} copies of the data structure that all threads can use to execute their operation, and uses a \emph{turn queue} for consensus~\cite{ramalhete2017poster}.
In the example, there are four available copies in the pool but the two threads only used three \texttt{Combined} instances.
In the scenario where a thread dies while executing on the second one (\texttt{c[1]}), this instance is no longer available to the remaining threads.
But in case the thread died after releasing the exclusive lock, then the instance would remain available to be used.
Also, the order of the mutations can be disposed off up until \texttt{curComb.head}, because at any given time there is an object referenced by \texttt{curComb}, that is protected by a shared lock, from which any thread can execute a copy.
In addition, there are at most \texttt{maxThread} operations after \texttt{curComb.head} that remain to be executed.

\subsection{Reader-Writer Lock with Strong Trylock}
\label{src:tryLock}

Access to each \texttt{Combined} instance and, consequently, to each copy of the object is managed by a reader-writer lock, \texttt{Combined.rwlock}.
In order to ensure wait-free progress, the reader-writer lock has to guarantee that from all the threads competing for the lock at least one will acquire it, a guarantee sometimes called \emph{deadlock freedom for trylock}, and furthermore the \texttt{trylock()} method must complete in a finite number of steps~\cite{correia2018strong}.

Based on these requirements, we chose to use the \textsf{StrongTryRWRI} reader-writer lock proposed in~\cite{correia2018strong}.
This lock's high scalability is capable of matching other state of the art reader-writer locks~\cite{calciu2013numa} while providing \texttt{downgrade()} functionality and strong trylock properties.
In addition, CX requires \emph{lock handover} between different threads when in shared mode.
In \CX, the \texttt{rwlock} of each \texttt{Combined} instance can be in one of four logical states:
unlocked; shared, \ie, read-only; exclusive, \ie, read-write; or handover.
The handover state, which is not typical in \texttt{rwlock} implementations, represents a state in which the lock is left in shared (read-only) mode without any thread actually using it, with the purpose of preventing writers from acquiring the lock in exclusive mode, \texttt{handoverLock()} or \texttt{downgradeToHandover()} will leave the lock in handover state.
The \texttt{rwlock} implementation allows the unlock of the shared mode by a different thread from the one which acquired the lock in shared mode, \texttt{handoverUnlock()}.



\subsection{Wait-Freedom}

The \texttt{applyUpdate()} method has only one loop where the number of iterations is not predetermined (line~\ref{alg:cxmut:apply:last-node}). 
For it to terminate, the traversal of the wait-free queue must encounter the node containing the process' update operation. 
The process starts by appending its update operation with sequence number $l$ to the wait-free queue and will proceed to acquire an exclusive lock for \texttt{Combined} instance \COMB{i}.
It is guaranteed by \autoref{prop:2} that each process will always have available two \texttt{Combined} instances to execute, even if all other threads fail holding two Combined instances locked, one in exclusive and another in shared mode.
For the process to execute the loop at line~\ref{alg:cxmut:apply:last-node}, \COMB{i}'s' state must be <$O_{i,j},\HEAD{i,j}$> where $l>j$, otherwise \texttt{applyUpdate()} would return at line~\ref{alg:cxmut:return-result-1}. 
From \autoref{prop:3}, \COMB{i}'s state will transition to <$O_{i,l},\HEAD{i,l}$> in $l-j$ iterations, unless a copy of object $O$ is required (line~\ref{alg:cxmut:copy:start}).
In case the process is unable to do a copy, it will return at line~\ref{alg:cxmut:return-result-2}, thus terminating the loop.
Otherwise, the copy from a \texttt{Combined} instance \COMB{k} referenced by \CURCOMB with state <$O_{k,m},\HEAD{k,m}$> is executed at line~\ref{alg:cxmut:copy-obj}.
The copy from $O_k$ is guaranteed to execute in a finite number of steps because \COMB{k} is protected by a shared lock, which guarantees that no update operation is taking place during the copy procedure.
After the copy is completed, \COMB{i} will be in the state <$O_{i,m},\HEAD{i,m}$>. 
The copy was performed from a \texttt{Combined} instance referenced by \CURCOMB, and \autoref{corollary:1} guarantees that $l-m \leq \texttt{maxThreads}$.
Consequently, the loop at line~\ref{alg:cxmut:apply:last-node} will iterate at most \texttt{maxThreads} times after the copy procedure.
In all possible scenarios, the loop will always iterate a finite number of steps.
The \texttt{applyUpdate()} method also calls \texttt{enqueue()} at line~\ref{alg:cxmut:enqueue-node} and the try-lock methods, \texttt{exclusiveTryLock()}, \texttt{exclusiveUnlock()}, \texttt{sharedTryLock()}, \texttt{sharedUnlock()}, \texttt{downgradeToHandover()} and \texttt{handoverUnlock()}. 
By definition, all these methods return in a finite number of steps, from which we conclude that \texttt{applyUpdate()} has wait-free progress.
In addition, assuming the sequential copy of an object is bounded, then the \texttt{applyUpdate()} method is also wait-free bounded.
The loop at line~\ref{alg:cxmut:apply:last-node} is the only one that can be unbounded.
Reclamation of the nodes is done once the thread-local circular buffer is full, which implies the queue is composed of a limited number of nodes. Considering that $size_\mathit{circbuff}$ represents the size of the circular buffer, then the maximum amount of iterations is bounded by:
\begin{equation}
size_\mathit{circbuff} \times maxThreads + maxThreads
\end{equation}
The \texttt{applyRead()} method iterates for a maximum of $\texttt{MAX\_TRIES}+\texttt{maxThreads}$. 
It calls \texttt{sharedTryLock()} and \texttt{sharedUnlock()}, which by definition return in a finite number of steps, resulting in wait-free progress. 
Refer to \autoref{sec:correctness} for the definitions of variables, Propositions and Corollary.


\section{Evaluation}
\label{sec:evaluation}

We now present a detailed evaluation of \CX and compare it with other state-of-the-art UCs and non-blocking data structures, using synthetic benchmarks.
Our microbenchmarks were executed on a dual-socket 2.10\;GHz Intel Xeon E5-2683 (``Broadwell'') with a total of 32 hyper-threaded cores (64 HW threads), running Ubuntu LTS and using \texttt{gcc} 7.2. 

\begin{figure*}[htbp!]
	\centering
	\includegraphics[scale=0.69,trim={25mm 12pt 0 12pt},clip]{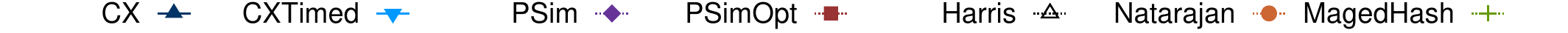}\\
	\hspace{-3mm}\includegraphics[scale=0.69]{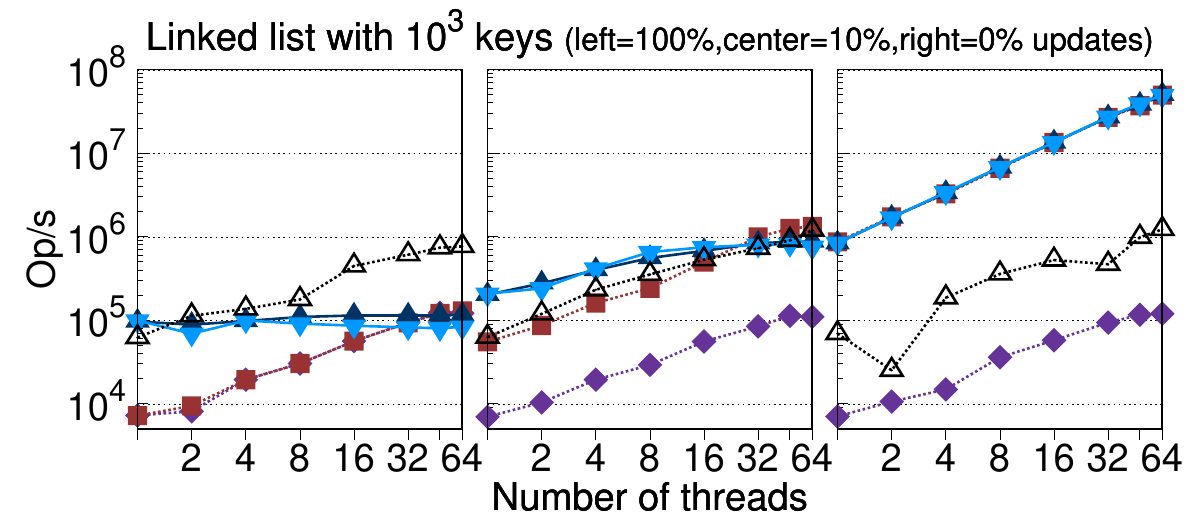}\hspace{2mm}\includegraphics[scale=0.69,trim={7mm 0 0 0},clip]{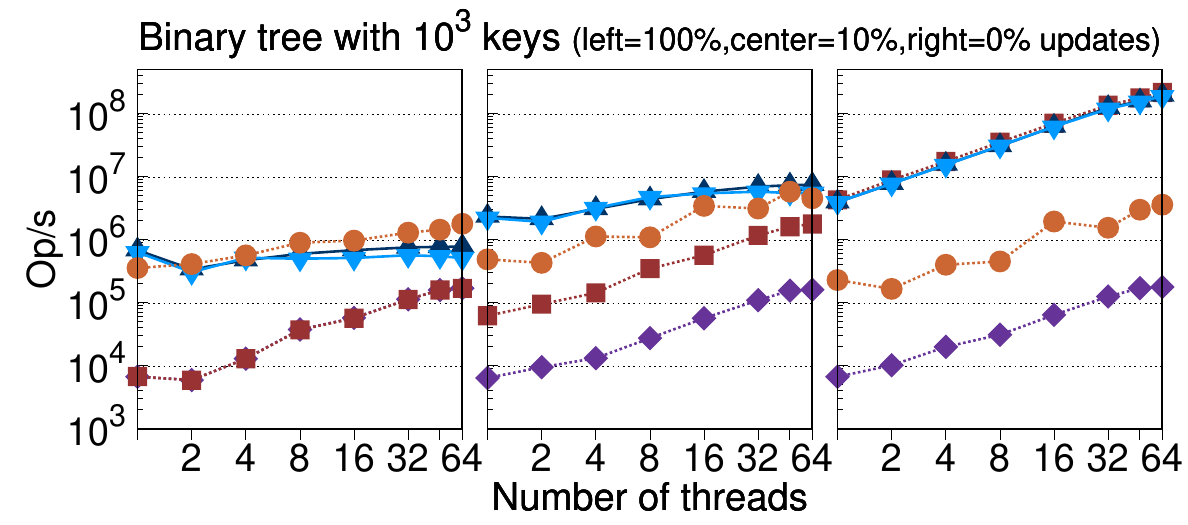}\\
	\hspace{-3mm}\includegraphics[scale=0.69]{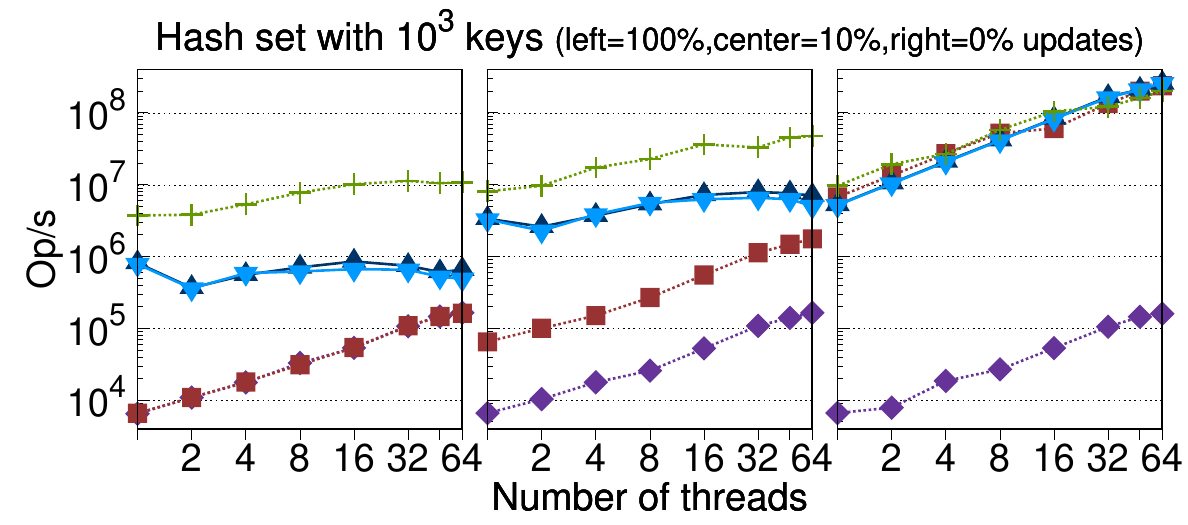}\hspace{2mm}\includegraphics[scale=0.69,trim={7mm 0 0 0},clip]{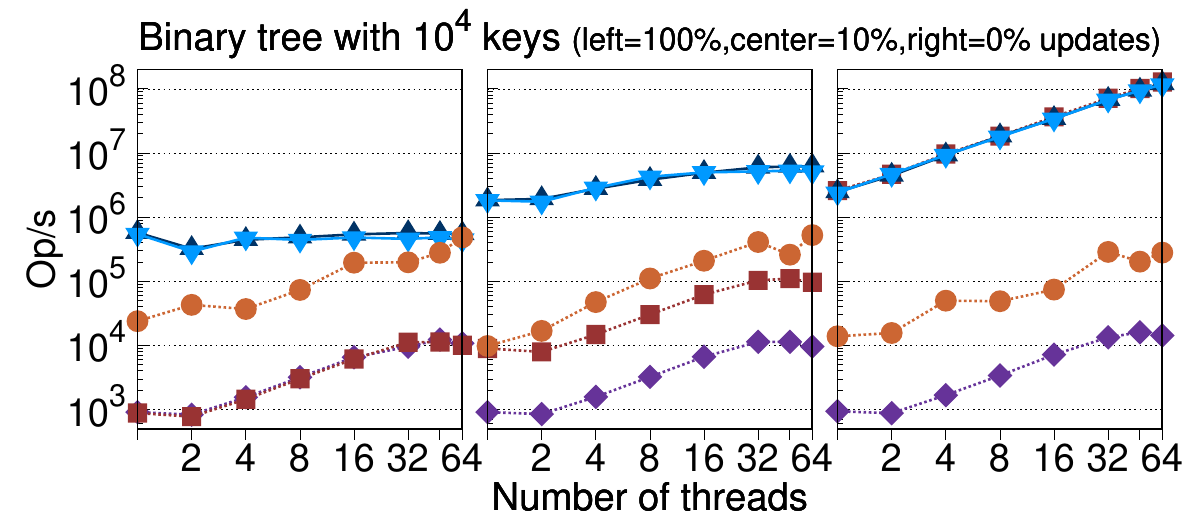}\\
	\hspace{-3mm}\includegraphics[scale=0.69]{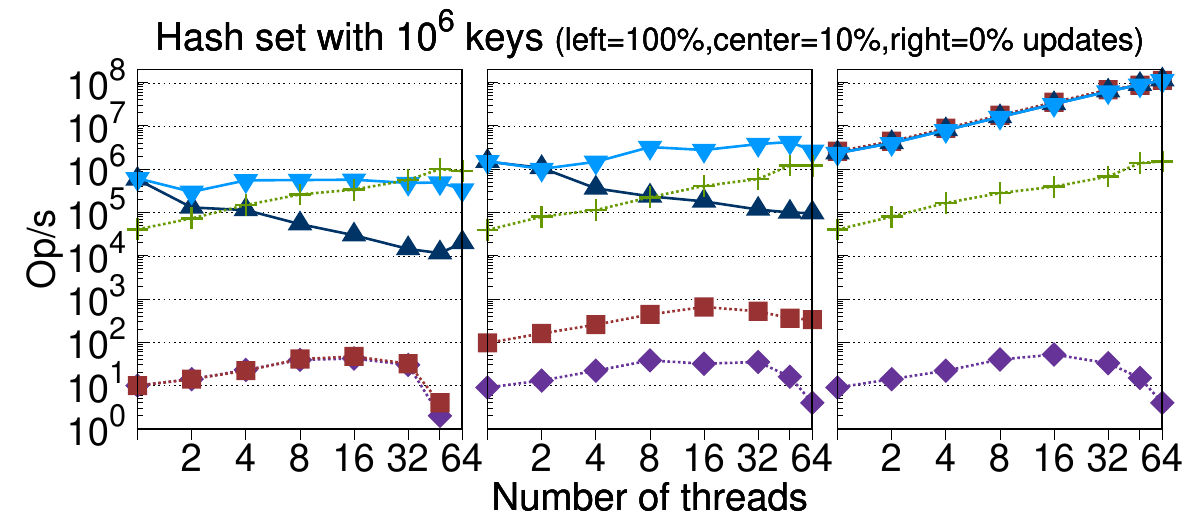}\hspace{2mm}\includegraphics[scale=0.69,trim={7mm 0 0 0},clip]{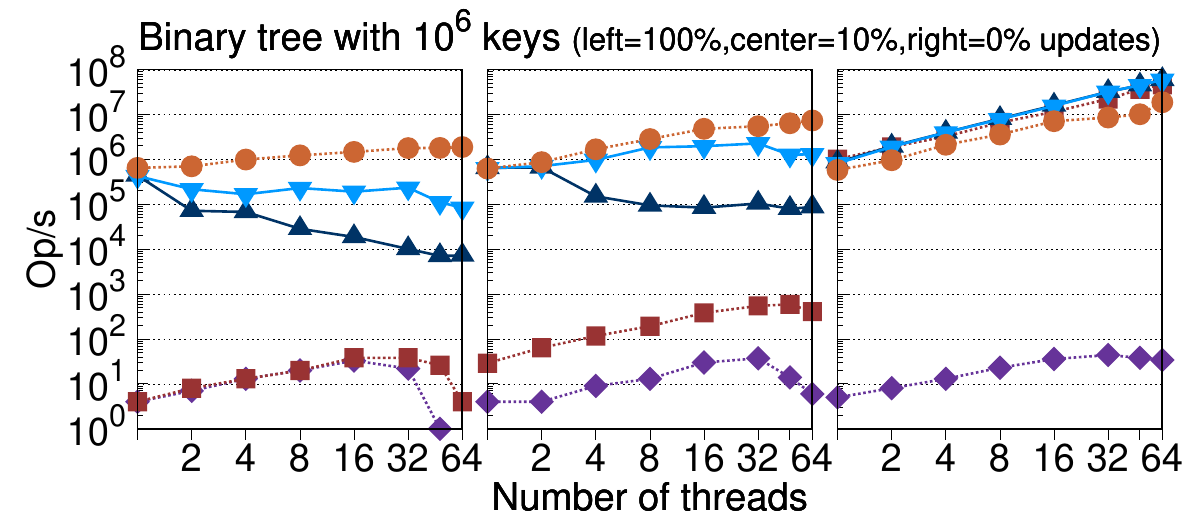}
	\caption{Left column (top to bottom): sets implemented using a linked list with $10^3$ keys and using a hash table with $10^3$ and $10^6$ keys. Right column: sets implemented using the balanced binary search trees std::set, with $10^3$, $10^4$ and $10^6$ keys. The results are presented with a logarithmic scale on both axes.}
	\label{fig:set-all}
\end{figure*}

Besides \CX, we used two other UCs which we now describe.
\texttt{PSim}~\cite{fatourou2011highly} is a UC with wait-free progress.
We adapted the original implementation available on \texttt{github}.
\texttt{PSimOpt} is an extension to \texttt{PSim}, where read-only operations have a different code path that allows them to scale, using a technique similar to the one we developed for \CX.
We also added a modified version of \CX, called \CXT, which restricts the amount of available \texttt{Combined} instances to four for a bounded period of time.
For \CXT, a thread is initially restricted to acquiring an exclusive lock on the first 4 \texttt{Combined} instances for a duration that corresponds to the time it takes to do a copy of the object.
After that amount of time has elapsed and its operation remains to be executed, the thread will acquire an exclusive lock on one of the 2$\times$\texttt{maxThreads} instances.
This can be seen as a blocking fast path with only 4 available \texttt{Combined} instances that can always revert to a slower path with wait-free progress. 
This approach further reduces the amount of object copies, because with high probability the first 4 instances are kept up to date.

Depending on the benchmark, we also compared with commonly available lock-free data structures.
\texttt{MagedHP} is a Harris linked list set modified by \citet{michael2002high}.
\texttt{Natarajan} is the relaxed tree by \citet{natarajan2014fast}.
\texttt{MagedHash} is the hash table by \citet{michael2002high}.
All implementations use Hazard Pointers~\cite{michael2004hazard}. 

All these data structures are \emph{sets} and the microbenchmarks described next have the same procedure.
A set is filled with 1,000 keys and we randomly select doing either a lookup or an update, with a probability that depends on the percentage of updates for each particular workload.
For a \emph{lookup}, we randomly select one key and call \texttt{contains(key)};
for an \emph{update}, we randomly select one key and call \texttt{remove(key)}, and if the removal is successful, we re-insert the same key with a call to \texttt{add(key)}, thus maintaining the total number of keys in the set (minus any ongoing removals).
Depending on the scenario, the procedure may be repeated for sets of different key ranges.
Each run takes 20 seconds, where a data point corresponds to the median of 5 runs.
All implementations will be available publicly.

The results of our experiments are shown in \autoref{fig:set-all}, with a log-log scale.
As expected, a sorted linked list protected by \CX is surpassed in most workloads by Maged-Harris' lock-free set because of the serialization of all operations in the wait-free queue necessary to reach consensus.
It is interesting to notice, however, that Maged-Harris algorithm is not able to outperform \CX in the scenario of $10\%$ updates.
\CX read operations do not require any pointer tracking during traversal because the data structure where the operation is executed is protected by a shared lock, which is not the case for traversals with Maged-Harris.

Let us now consider experiments with hash sets.
The \texttt{MagedHash} algorithm uses a pre-allocated array of 1,000 buckets and its advantage over \CX is significant, due to \CX serializing all mutative operations, while updates on the \texttt{MagedHash} are mostly disjoint.
However, when we insert one million keys, the fact that there are only 1,000 buckets causes increased serialization of the update operations, giving \CX an edge in nearly all scenarios.
Currently there is no known efficient hand-made lock-free resizable hash set with lock-free memory reclamation.

Regarding balanced trees, three different workloads were executed with $10^3$, $10^4$ and $10^6$ keys, shown in \autoref{fig:set-all}.
Natarajan's tree is not shown for $10^6$ keys because it takes two hours to fill up, making it unsuitable for such a scenario.
This occurs because it is a \emph{non-balanced} tree and our benchmarks execute a consecutive fill of the keys, causing this tree to effectively become a linked list of nodes because it is \emph{never} rebalanced.
There are a few lock-free balanced trees in the literature~\cite{brown2014general}, however, there is no known implementation with hazard pointers or any other lock-free memory reclamation.  
Balanced trees like the \texttt{std:set} sequential implementation we use in CX do not suffer from these issues.
For a small tree with $10^3$ keys, with $100\%$ updates, \CX is the most efficient for single-threaded execution and is not far behind the lock-free tree for the remaining thread counts.
As the ratio of read-only operations increase, \CX improves and at $10\%$ updates it is able to beat the lock-free tree, irrespective of the number of threads.
For a tree with $10^4$ and $10^6$ keys, \CX has the advantage on all tested scenarios.

The two \CX implementations evaluated in this section give high scalability for read-mostly workloads regardless of the underlying data structure.
Read-only operations in \CX can almost always acquire the shared lock after a few trylock attempts, which implies that the synchronization cost is a few sequentially consistent stores.
This high throughput surpasses equivalent lock-free data structures, while providing wait-free progress and linearizable consistency for any operation.
For high update workloads, equivalent lock-free data structures may have higher performance than \CX.

As for other UCs, \texttt{PSim} drags far below in all tested scenarios due to serializing all operations, even though it has been until now the best of the non-interposing UCs, easily surpassing Herlihy's original wait-free UC (not shown in this paper).
Our optimized implementation with scalable reads, \texttt{PSimOpt}, greatly improves the throughput on workloads with $0\%$ updates but as soon as the number of update operations increase it shows similar performance when compared with \texttt{PSim}.
Source code for CX and benchmarks is available here \url{https://github.com/pramalhe/CX}

\subsection{Memory Usage}

For large data structures, the amount of memory required to execute the program can be a determining factor when choosing a more suitable concurrency synchronization.
We conducted an experiment meant to evaluate the trade-off between memory usage and throughput.
The experiment follows the same procedure as for update-only workloads, using the balanced binary search tree available in STL, with pre-filled trees of 1 and 10 million keys.
The maximum memory usage is measured executing the same microbenchmark.
Each data point of \autoref{fig:memory-usage} is the highest value of two runs, each run executing for 100 seconds.

\begin{figure*}[htbp!]
	\centering
	\includegraphics[scale=0.72]{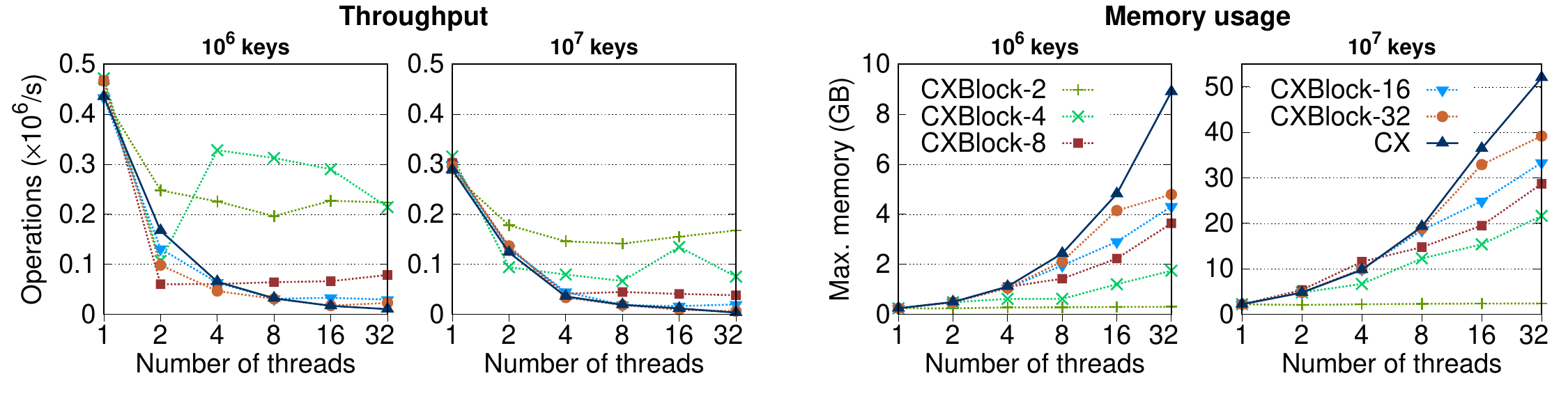}
	\caption{Throughput (left) vs. maximum memory usage (right) with 100\% updates for a pre-filled \texttt{std::set} of $10^6$ and $10^7$ keys. \texttt{CXBlock-k} represents the \CX universal construct where $maxObjs=k$. \CX is the wait-free UC where the maximum number of objects is 64.}
	\label{fig:memory-usage}
\end{figure*}



We observe in \autoref{fig:memory-usage} that, unsurprisingly, the configuration with the lowest memory usage is when \CX has $maxObjs$ set to 2.
It is only using two \texttt{Combined} instances with a constant maximum memory usage around 280\;MB.
When increasing the size of the data structures from $10^6$ to $10^7$ keys, we observe that memory usage grows as expected by a factor of ten, around 2\;GB.

On the other hand, \CX with wait-free progress (\ie, when $maxObjs$ is set to 64) is the configuration with the highest memory requirements.
As the number of threads grows, we observe both an increase in memory usage and a decrease in throughput, but this is compensated by the additional guarantees of resilience in case of thread failures.
Our experiments also show that, for the tree with 1 million keys, \texttt{CXBlock-4} sometimes achieves better performance than \texttt{CXBlock-2}.
We can reach the conclusion that, if the application can relax the progress guarantees for update operations, then a suitable configuration would be to use up to 4 object instances.
This would provide a good trade-off between memory usage, throughput and progress.

\section{Conclusion}
\label{sec:conclusion}

The appeal of generic techniques like wait-free universal constructs (UC) stems from the difficulty in designing \emph{hand-written} non-blocking data structures.
These constructs can transform any sequential implementation of a data structure into a correct wait-free data structure, with linearizable consistency for \emph{all} operations.

\CX is the first non-instrumenting UC capable of transforming \emph{any} sequential implementation of a data structure with unforeseen method implementations to be considered for multi-threaded applications with performance that rivals and surpasses hand-made lock-free implementations.

Moreover, \CX has integrated wait-free memory reclamation, a feature that most hand-written lock-free data structures do not provide.
Using CX we have implemented the first wait-free binary \emph{balanced} tree, showing that CX makes it possible to create new data structures for which no hand-made counterparts exist yet.
\CX's huge leap in performance compared with previous UCs is due to the significant reduction of copy operations, where available copies are instead reused and updated.

\bibliographystyle{ACM-Reference-Format}
\bibliography{references}
\pagebreak
\newpage

\appendix

\section{Correctness}
\label{sec:correctness}

\newcommand{\OP}[1]{\ensuremath{\mathit{op}_{#1}}\xspace}
\newcommand{\OPINV}[1]{\ensuremath{\mathit{op}_{#1}.\mathit{inv}}\xspace}
\newcommand{\OPRES}[1]{\ensuremath{\mathit{op}_{#1}.\mathit{res}}\xspace}

We discuss the correctness of \CX using standard definitions and notations for linearizability~\cite{linearizability}, which we briefly recall below.
A concurrent execution is modeled by a history, \ie, a sequence of events.
Events can be operation invocations and responses, denoted respectively as \OPINV{} and \OPRES{}.
Each event is labeled with the process and with the object $O$ to which it pertains.
A subhistory of a history $H$ is a subsequence of the events in $H$.
A response matches an invocation if they are performed by the same process on the same object.
An operation in a history $H$ consists of an invocation and the next matching response.
An update operation may cause a change of state in the object, with a visible effect to other processes, while a read-only operation has no effects visible to other processes.
An invocation is pending in $H$ if no matching response follows it in $H$.
An extension of $H$ is a history obtained by appending responses to zero or more pending invocations in $H$, and $\mathit{complete}(H)$ denotes the subhistory of H containing all matching invocations and responses.
All references to specific lines of code in this section refer to \autoref{alg:cxmut}.

\begin{defi}
\emph{(Happens before)}
if $\OPRES{1} <_\textsf{hb} \OPINV{2}$ then $\OP{1} <_\textsf{hb} \OP{2}$.
\end{defi}
\begin{defi}
\emph{(Subhistory)}
Given a history $H$, a subhistory $S$ of $H$ is such that if \OP{2} belongs to $S$
and $\OP{1} <_\textsf{hb} \OP{2}$ in $H$, then \OP{1} belongs to $S$ and $\OP{1} <_\textsf{hb} \OP{2}$ in $S$.
\end{defi}
\begin{defi}
\emph{(Partial subhistory)}
Given a history $H$ and update operations \OP{1} and \OP{2}, a partial subhistory $S$ of $H$ is such that if \OP{2} belongs to $S$ and $\OP{1} <_\textsf{hb} \OP{2}$ in $H$, then \OP{1} belongs to $S$ and $\OP{1} <_\textsf{hb} \OP{2}$ in $S$.
\end{defi}
\begin{defi}
\emph{(Linearizability)}
A history $H$ is linearizable if $H$ has an extension $H'$ and there is a legal sequential subhistory $S$ such that:
\begin{enumerate*}[label=\emph{(\roman*)}]
\item $\mathit{complete}(H')$ is equivalent to $S$; and
\item if an operation $\OP{1} <_\textsf{hb} \OP{2}$ in $H$, then the same holds in $S$.
\end{enumerate*}
\end{defi}
\CX's correctness relies on a linearizable wait-free queue, and a linearizable reader-writer lock with strong guarantee for trylock methods.
The queue represents the sequence of update operations applied to the object $O$, establishing the partial history of the concurrent execution. 
For simplicity, we assume that the queue is formed by a sequence of nodes and each has a unique sequence number, which is monotonically increasing for consecutive nodes.

\CX requires a linearizable wait-free enqueue method and a linearizable traversal of the queue that provides a partial subhistory of the history $H$.
Regarding the linearizable reader-writer lock, it must provide strong guarantees for the \texttt{sharedTrylock()} and \texttt{exclusiveTrylock()} methods. 
Both methods must satisfy the property of \emph{deadlock-freedom}, \ie, the critical section will not become inaccessible to all processes, and their invocation must complete in a finite number of steps.

We denote by \COMB{i} the $i_\mathit{th}$ \texttt{Combined} instance of the \texttt{combs[]} array.
At any given moment, \COMB{i} is in a state represented by the pair <$O_{i,j},\HEAD{i,j}$>.
$O_{i,j}$ represents the $i_\mathit{th}$ simulation of object $O$ where $j$ corresponds to the sequence number in a node of the wait-free queue. 
As such, $O_{i,0}$ is the initialized object and $O_{i,j}$ is the simulation of object $O$ after sequentially applying all update operations up to the operation with sequence number $j$. 
We assume all \COMB{i} instances start in the initial state <$O_{i,0},\HEAD{i,0}$> where $\HEAD{i,0}$ is the sentinel node of the wait-free queue.
$\HEAD{i,j}$ represents a node with sequence number $j$ in the wait-free queue.
$O_{i,j}.\OP{j+1}()$ represents the execution of operation $\OP{j+1}()$ on $O_{i,j}$, and the resulting object $O_{i,j+1}$ will contain the effects of $\OP{j+1}()$.
We define \CURCOMB as the \texttt{Combined} instance on which read-only operations execute.

\begin{prop}
\label{prop:1}
\CURCOMB can only transition between different \texttt{Combined} instances both protected by a shared lock.
\end{prop}
\begin{proof}
At the start of the execution, \CURCOMB references a \texttt{Combined} instance protected by a shared lock.
The state transition of \CURCOMB occurs in line~\ref{alg:cxmut:cas} between two \texttt{Combined} instances, referred as $lComb$ and $newComb$.
The lock associated with the $newComb$ instance is acquired in exclusive mode in line~\ref{alg:cxmut:lock-exclusive} and is later downgraded to shared mode in line~\ref{alg:cxmut:downgrade-to-shared-lock}. 
It is not possible for $newComb$ to be the $lComb$ because $lComb$ is protected by a shared lock.
From the moment a process $q$ acquires the exclusive lock protecting $newComb$, until the state transition in line~\ref{alg:cxmut:cas}, other processes may change \CURCOMB from $lComb$ to reference another \texttt{Combined} instance. 
However, those processes will not be able to change \CURCOMB to reference $newComb$.
Any other process attempting to transition \CURCOMB will have to first acquire an exclusive lock on a \texttt{Combined} instance, and it is impossible that this instance is $newComb$ because $newComb$'s lock is held by process $q$.
As such, \CURCOMB can only transition to a different \texttt{Combined} instance and that instance's lock is held in shared mode.
\end{proof}

\begin{prop}
\label{prop:2}
At most $2\times\texttt{maxThreads}$ \texttt{Combined} instances are necessary to guarantee that an update operation will acquire an exclusive lock on one of the \texttt{Combined} instances.
\end{prop}

\begin{proof}
A process executing \texttt{applyUpdate()} will require at most two \texttt{Combined} instances at any given time, the acquisition of an exclusive lock at line~\ref{alg:cxmut:lock-exclusive} and a shared lock at line~\ref{alg:cxmut:acquire-shared} or~\ref{alg:cxmut:lock-shared}.
Assuming the reader-writer trylock methods guarantee that no available \texttt{Combined} instance can remain inaccessible to all competing processes, this implies that any process that failed to acquire a lock in a \texttt{Combined} instance is sure that the instance is in use by a competing process.
By induction, lets consider that processes $q_1, \ldots, q_{n-1}$ use $2\times(\texttt{maxThreads}-1)$ \texttt{Combined} instances. The last process $q_n$ will have available the last two \texttt{Combined} instances.
Considering that process $q_1$ releases the shared lock of \COMB{i} and leaves the other in handover state.
In a subsequent call to \texttt{applyUpdate()}, process $q_1$ will acquire the exclusive lock on \COMB{i} because this is the first available \texttt{Combined} instance when traversing the \texttt{combs[]} array, and the shared lock will be acquired on one of the $2\times(\texttt{maxThreads}-1)$ \texttt{Combined} instances that can potentially be \CURCOMB.
In the event that process $q_n$ transitions \CURCOMB to one of its two available instances, then process $q_1$ can acquire that \texttt{Combined} instance in shared mode but it will leave a precedent \texttt{Combined} instance available to be acquired in exclusive mode by process $q_n$.
This shows that there will always be two \texttt{Combined} instances available to process $q_n$, the maximum it may need.
\end{proof}

\begin{prop}
\label{prop:3}
For any \COMB{i}, an update operation with sequence $l$ will transition atomically from <$O_{i,j},\HEAD{i,j}$> to <$O_{i,l},\HEAD{i,l}$>.
\end{prop}

\begin{proof}
Every \texttt{Combined} instance \COMB{i} is protected by a reader-writer trylock, granting exclusive access in line~\ref{alg:cxmut:lock-exclusive} by proposition 2 to only one process. This process will be allowed to mutate its state from the pair <$O_{i,j},\HEAD{i,j}$> to a subsequent state.
Only a process that is executing an update operation can acquire an exclusive lock on \COMB{i}. Its update operation was previously appended to the queue and we assume the sequence number of the operation is $l$.
Subsequently, the process will execute the statements from line~\ref{alg:cxmut:apply:start} to~\ref{alg:cxmut:apply:end}, where the initial simulated object is $O_{i,j}$.
$O_{i,j}$ will be subjected to the execution of the sequence of operations $op_{k}()$ where $k=j+1,...,l$, transitioning the simulated object to $O_{i,l}$.
The traversal of the queue is required to be linearizable.
The sequence of operations observed by a process traversing the queue is the same for all other processes.
All concurrent mutative operations applied to object $O$ were previously appended to the queue, and the queue establishes the partial history of the concurrent execution.
The simulated object has now mutated to $O_{i,l}$ and in line~\ref{alg:cxmut:update-head-2}, \COMB{i} state will transition to <$O_{i,j},\HEAD{i,j}$> where \HEAD{i,l} represents the node containing the last operation applied to the simulated object $O_{i,l}$.
Only after the transition to state <$O_{i,l},\HEAD{i,l}$> is completed, will the \texttt{Combined} instance \COMB{i} be made available to other processes, in line~\ref{alg:cxmut:downgrade-to-shared-lock}.
\end{proof}

\begin{prop}
\label{prop:4}
\CURCOMB always transitions from \COMB{i} to \COMB{k}, with respective states <$O_{i,j},\HEAD{i,j}$> and <$O_{k,l},\HEAD{k,l}$>, where $i \neq k$ and $l>j$.
\end{prop}
\begin{proof}
\autoref{prop:4} follows from \autoref{prop:1} and~\autoref{prop:3}.
In addition, the state transition can only occur (line~\ref{alg:cxmut:cas}) if $\HEAD{i,j}$ does not satisfy the condition at line~\ref{alg:cxmut:compare-ticket}.
This condition guarantees that $l>j$.
\end{proof}

\begin{lem}
\label{lem:1}
An update operation with sequence number $l$ can only return, after ensuring $curComb$ transitions to a $Comb_i$ with state $<O_{i,j},\HEAD{i,j}>$ where $l \leq j$.
\end{lem}

\begin{proof}
An update operation with sequence number $l$ will complete as soon as \texttt{applyUpdate()} returns (lines~\ref{alg:cxmut:return-result-1}, \ref{alg:cxmut:return-result-2}, \ref{alg:cxmut:return-result-3} or \ref{alg:cxmut:return-result-4}).
After the exclusive lock is acquired for the \texttt{Combined} instance \COMB{k} with state <$O_{k,m},\HEAD{k,m}$> at line~\ref{alg:cxmut:compare-ticket-head}, it will validate if the sequence number of $\HEAD{k,m}$ is greater than or equal to $l$, implying that $l \leq m$.
If \COMB{k} can have been acquired in exclusive mode, then any previous update operation with sequence number $m$ that updated \COMB{k} with update operations until $\HEAD{k,m}$ had to guarantee that \CURCOMB was referencing a \texttt{Combined} instance \COMB{i} with state <$O_{i,j},\HEAD{i,j}$> where $m \leq j$.
This proves that an update operation with sequence $l$ can return at line~\ref{alg:cxmut:return-result-1} with the guarantee that \CURCOMB was at least referencing a \texttt{Combined} instance \COMB{i} with state <$O_{i,j},\HEAD{i,j}$> where $l \leq j$.

Execution from lines~\ref{alg:cxmut:copy-required} to~\ref{alg:cxmut:copy-completed} occurs when the \texttt{Combined} instance acquired in exclusive mode requires a copy of \CURCOMB.
The method \texttt{getCombined()} can only return null in two cases:
in case \CURCOMB sequence $j$ is higher than or equal to $l$; otherwise, if after \texttt{maxThreads} trials it fails to acquire the shared lock of the current \CURCOMB, represented as a specific \COMB{i} with state <$O_{i,j},\HEAD{i,j}$>.
This can only occur if \CURCOMB changed at least \texttt{maxThreads} times.
By \autoref{prop:4}, \CURCOMB must be referencing a \texttt{Combined} instance with state <$O_{k,m},\HEAD{k,m}$> where $j+\texttt{maxThreads} \leq m$.
Assuming no operation can return before ensuring its operation is visible at \CURCOMB then $l \leq j+\texttt{maxThreads}$, implying that $l \leq m$. 
This proves that after \texttt{maxThreads} transitions of \CURCOMB it must contain the update operation with sequence $l$.

On line~\ref{alg:cxmut:return-result-3} the update operation returns after ensuring that \CURCOMB references a \texttt{Combined} instance with sequence number lower than $l$ (line~\ref{alg:cxmut:compare-ticket}) and successfully transitions \CURCOMB, which by \autoref{prop:4} will map to a \texttt{Combined} instance with state <$O_{k,l},\HEAD{k,l}$> where $i \neq k$ and $l>j$.

When the update operation returns after \texttt{maxThreads} failed attempts to acquire the shared lock of the current \CURCOMB on line~\ref{alg:cxmut:return-result-4}, \CURCOMB must contain the update operation with sequence $l$.
\end{proof}

From \autoref{lem:1}, we can directly infer the following corollary.

\begin{cor}
\label{corollary:1}
\CURCOMB always references a \texttt{Combined} instance with state $<O_{i,j},\HEAD{i,j}>$ where $l-j \leq \texttt{maxThreads}$, with $l$ the sequence number at the tail of the wait-free queue.
\end{cor}

We now introduce the remaining lemmas that will allow us to prove linearizability of the \CX universal construct.

\begin{lem}
\label{lem:2}
Given \OP{1} and \OP{2} two update operations on object $O$, if $\OP{1} <_\textsf{hb} \OP{2}$ and \OP{2} belongs to $S$ then \OP{1} belongs to $S$, where $S$ is a subhistory of $H$ on $O$.
\end{lem}
\vspace{-6pt}
\begin{proof}
Follows from proposition 3.
\end{proof}
\begin{lem}
\label{lem:3}
Given \OP{u} an update operation and \OP{r} a read-only operation, both on object $O$, if $\OP{u} <_\textsf{hb} \OP{r}$ then \OP{r} will have to see the effects of \OP{u} on $O$.
\end{lem}
\begin{proof}
If \OP{r} execution accesses a \texttt{Combined} instance that does not contain the effects of \OP{u}, then \CURCOMB has not yet transitioned to an instance that contains \OP{u}. By \autoref{lem:1}, the \OP{u} is only considered to take effect after \CURCOMB transitions to a Combined instance which $O_i$ contains the effects of \OP{u}.
This means that \OP{r} could take place before \OP{u}, implying \OP{r} $<_\textsf{hb}$ \OP{u}, thus contradicting the initial assumption.
\end{proof}
\begin{lem}
\label{lem:4}
Given \OP{r} a read-only operation and \OP{u} an update operation, both on object $O$, if $\OP{r} <_\textsf{hb} \OP{u}$ then \OP{r} will not see the effects of \OP{u} on $O$.
\end{lem}
\begin{proof}
\OP{u} can only return after guaranteeing that \CURCOMB has transitioned to a \texttt{Combined} instance that contains the effects of \OP{u}, by \autoref{lem:1}.
Any read operation accesses only the Combined instance referenced by \CURCOMB.
If \OP{r} accesses a Combined instance that contains the effects of \OP{u} then it would be possible to consider as if \OP{u} had occurred before \OP{r}, because the current \CURCOMB already contains \OP{u}. This contradicts the definition of happens-before, if the response of \OP{u} can occur before the invocation of \OP{r} then $\OP{u} <_\textsf{hb} \OP{r}$ which contradicts the initial assumption. Implying \CURCOMB can not contain the effects of \OP{u} and, therefore, \OP{r} will not see the effects of \OP{u} over object $O$.
\end{proof}
\begin{lem}
\label{lem:5}
Given \OP{1} and \OP{2} two identical read-only operations on object $O$, if $\OP{1} <_\textsf{hb}\OP{2}$ then \OP{2} returns the same result as \OP{1}, unless an update operation \OP{u} interleaves.
\end{lem}
\begin{proof}
By \autoref{lem:1}, only update operations can transition the state of \CURCOMB.
All read operations access only the Combined instance referenced by \CURCOMB.
If there is no update operation \OP{3} interleaving between \OP{1} and \OP{2}, the read operations will necessary access the same Combined instance yielding the same result.
\end{proof}

The proof of \autoref{lem:3}, \ref{lem:4} and~\ref{lem:5} rely on the fact that any update operation must guarantee that \CURCOMB contains the effects of its operation, by \autoref{lem:1}, and that the read operation always executes on the object referenced by \CURCOMB.

\begin{thm}
\label{theorem:1}
The \CX universal construct provides linearizable operations.
\end{thm}

\begin{proof}
Follows from \autoref{lem:2}, \ref{lem:3}, \ref{lem:4} and~\ref{lem:5}.
\end{proof}

Processes calling \texttt{applyUpdate()} are imposed a FIFO linearizable order by the wait-free queue, forcing each process to see the mutations to be applied in the same global order, and if needed, apply these mutations on its local data structure copy \texttt{Combined.obj}.
This means that the linearization point from writers to writers is the enqueuing in the wait-free queue on line~\ref{alg:cxmut:enqueue-node}.
For writers to readers, the mutation becomes visible when \emph{published} in \texttt{curComb}.
As such, the linearization point from writers to readers is the CAS in \texttt{curComb} on line~\ref{alg:cxmut:cas}, which makes the mutation visible to readers on the load of line~\ref{alg:cxread:read:curcomb-load}.

\end{document}